\documentclass{sig-alternate}

\usepackage{amsthm,amssymb,amsmath}
\usepackage{graphics}
\usepackage{tikz}
\usepackage[plainpages=false,pdfpagelabels,colorlinks=true,citecolor=blue,hypertexnames=false]{hyperref}
\usepackage{color}

\newtheorem{theorem}{Theorem}
\newtheorem{prop}[theorem]{Proposition}

\newtheorem{lemma}[theorem]{Lemma}
\newtheorem{algorithm}[theorem]{Algorithm}

\newtheorem{defi}[theorem]{Definition}
\newtheorem{example}[theorem]{Example}

\def\qed{\quad\rule{1ex}{1ex}}

\def\ord{\operatorname{ord}}

\def\O{\mathrm{O}}
\def\k{p}

\let\set\mathbb

\def\lclm{\operatorname{lclm}}

\def\<#1>{\langle#1\rangle}

\newfont{\mycrnotice}{ptmr8t at 7pt}
\newfont{\myconfname}{ptmri8t at 7pt}
\let\crnotice\mycrnotice%
\let\confname\myconfname%

\permission{Permission to make digital or hard copies of all or part of this work for personal or classroom use is granted without fee provided that copies are not made or distributed for profit or commercial advantage and that copies bear this notice and the full citation on the first page. Copyrights for components of this work owned by others than ACM must be honored. Abstracting with credit is permitted. To copy otherwise, or republish, to post on servers or to redistribute to lists, requires prior specific permission and/or a fee. Request permissions from permissions@acm.org.}

\conferenceinfo{ISSAC'15,}{July 6--9, 2015, Bath, United Kingdom.\\
{\mycrnotice{Copyright is held by the owner/author(s). Publication rights licensed to ACM.}}}
\copyrightetc{ACM \the\acmcopyr}
\crdata{978-1-4503-3435-8/15/07\ ...\$15.00.\\
Update with your DOI string http://XXXX provided in ACM eform confirmation email \& pdf of ACM eform}

\conferenceinfo{ISSAC'15,}{July 6--9, 2015, Bath, United Kingdom} 
\copyrightetc{Copyright is held by the owner/author(s). Publication rights licensed to ACM. \\ \the\acmcopyr}
\crdata{ACM 978-1-4503-3435-8/15/07\ ...\$15.00}
\doi{DOI: http://dx.doi.org/10.1145/2755996.2756658}

\CopyrightYear{2015}

\begin{document}

\allowdisplaybreaks

\title{Integral D-Finite Functions}

\numberofauthors{2}

\author{%
 \alignauthor
 \leavevmode
 \mathstrut Manuel Kauers\titlenote{Supported by the Austrian Science Fund (FWF): Y464.\vspace{-4pt}}\\[\smallskipamount]
  \affaddr{\leavevmode\mathstrut RISC / Johannes Kepler University}\\
  \affaddr{\leavevmode\mathstrut 4040 Linz, Austria}\\
  \affaddr{\leavevmode\mathstrut mkauers@risc.jku.at}
 \and
 \mathstrut Christoph Koutschan\titlenote{Supported by the Austrian Science Fund (FWF): W1214.\vspace{5pt}}\\[\smallskipamount]
  \affaddr{\mathstrut RICAM / Austrian Academy of Sciences}\\
  \affaddr{\mathstrut 4040 Linz, Austria}\\
  \affaddr{\mathstrut christoph.koutschan@ricam.oeaw.ac.at}
}

\maketitle
\begin{abstract}
  We propose a differential analog of the notion of integral closure of
  algebraic function fields. We present an algorithm for computing the integral
  closure of the algebra defined by a linear differential operator. Our
  algorithm is a direct analog of van Hoeij's algorithm for computing integral
  bases of algebraic function fields.
\end{abstract}


\category{I.1.2}{Computing Methodologies}{Symbolic and Algebraic Manipulation}[Algorithms]


\terms{Algorithms}


\keywords{Integral Basis, D-finite Function, Differential Operator}


\section{Introduction}

The notion of integrality is a classical concept in the theory of algebraic
field extensions. If $R$ is an integral domain, $k$ the quotient field of~$R$, and
$K=k(\alpha)$ an algebraic extension of~$k$ of degree~$d$, then an element of~$K$ is
called \emph{integral} if its monic minimal polynomial has coefficients
in~$R$. While $K$ forms a $k$-vector space of dimension~$d$, the set of
all integral elements of $K$ forms an $R$-module, called the \emph{integral
 closure} (or \emph{normalization}) of $R$ in~$K$, and commonly denoted by $\mathcal{O}_K$. A $k$-vector
space basis of~$K$ which at the same time generates $\mathcal{O}_K$ as an
$R$-module is called an \emph{integral basis.} For example, when $R=\set Z$,
$k=\set Q$, and $K=\set Q(\alpha)$ with $\alpha=\sqrt[3]4$, then the canonical
vector space basis $\{1,\alpha,\alpha^2\}$ of~$K$ is not an integral basis,
because $\tfrac12\alpha^2=\sqrt[3]{2}$ is an integral element of~$K$ (its
minimal polynomial is $x^3-2$) but not a $\set Z$-linear combination of
$1,\alpha,\alpha^2$. An integral basis in this example is
$\{1,\alpha,\tfrac12\alpha^2\}$.


The concept of integral closure has been studied in rather general
  domains~\cite{Swanson06,deJong98}.  To compute an integral basis for an
  algebraic number field, special algorithms have been
  developed~\cite{Ford78,cohen93}.  At least two different approaches are
known for algebraic function fields, i.e., the case when $R=C[x]$ for some
field~$C$, $k=C(x)$, and $K=k[y]/\<M>$ for some irreducible polynomial $M\in
k[y]$.  The algorithm derived by Trager~\cite{Trager84} in his thesis is an
adaption of an algorithm for number fields, and the algorithm by van
Hoeij~\cite{vanHoeij94} is based on the idea of successively canceling lower
order terms of Puiseux series.

The theory of algebraic functions parallels in many ways the theory of D-finite
functions, i.e., the theory of solutions of linear differential operators. It is
therefore natural to ask what corresponds to the notion of integrality in this
latter theory. In the present paper, we propose such a definition and give an
algorithm which computes integral bases according to this definition. Our
algorithm and the arguments underlying its correctness are remarkably similar to
van Hoeij's algorithm for computing integral bases of algebraic function fields.

In view of the key role that integral bases play for indefinite integration
(Hermite reduction) of algebraic
functions~\cite{Trager84,bronstein98,bronstein98a}, we have hope that results
presented below will help to develop new algorithms for indefinite integration
of D-finite functions. An example pointing in this direction is given in the end.

\medskip\noindent
\textbf{Acknowledgment.}
We want to thank the anonymous referees for their detailed and valuable comments.

\section{Integral Functions,\texorpdfstring{\hfill\break}{} Integral Closure, and\texorpdfstring{\hfill\break}{} Integral Bases}

Throughout this paper, let $C$ be a computable field of characteristic zero,
$\bar C$~an algebraically closed field containing~$C$ (not necessarily the
smallest), and $x$~transcendental over~$\bar C$.  When $R$ is a subring of $\bar
C(x)$, we write $R[D]$ for the algebra of differential operators with
coefficients in~$R$, i.e., the algebra of all (formal) polynomials
$\ell_0+\ell_1D+\cdots+\ell_rD^r$ with $\ell_0,\dots,\ell_r\in R$.  This algebra
is equipped with the natural addition and the unique noncommutative
multiplication respecting the commutation rules $Dc=cD$ for all $c\in R\cap\bar
C$ and $Dx=xD+1$. Typical choices of $R$ will be $C[x]$, $\bar C[x]$, $C(x)$, or
$\bar C(x)$ in the following.

For an operator $L=\ell_0+\ell_1D+\cdots+\ell_rD^r\in\bar C[x][D]$ with $\ell_r\neq0$
we denote by $\ord(L)=r$ the order of~$L$. Recall that such an operator with $x\nmid\ell_r$
admits a fundamental system of formal power series solutions, i.e., the
vector space $V\subseteq\bar C[[x]]$ consisting of all the power series $f$ with
$L\cdot f=0$ has dimension~$r$. When $x\mid\ell_r\neq0$, there is still a
fundamental system of generalized series solutions of the form
$\exp(p(x^{-1/s}))x^\nu a(x^{1/s},\log(x))$ for some $s\in\set N$, $p\in\bar
C[x]$, $\nu\in\bar C$, $a\in\bar C[[x]][y]$.  (This notation is not meant to
imply that $a$ has a nonzero constant term, so the series in general does not
start at~$x^\nu$ but at some $x^{\nu+i}$ where $i\in\set N$ is such that $x^i$ is
the lowest order term of~$a$.)  We restrict our attention here to the case where
$p=0$ and $s=1$. Moreover, we want to assume that $\nu\in C$ (this can always
be achieved by a suitable choice of~$C$), to ensure that the output of our
algorithm involves only coefficients in~$C$. Hence we only consider operators~$L$
which admit a fundamental system in $\bigcup_{\nu\in C} x^\nu\bar C[[x]][\log x]$.
It is well known~\cite{ince26} how to determine the first terms of a basis of such
solutions for a given operator $L\in\bar C[x][D]$.  By a linear change of
variables, the same techniques can also be used to find the first terms of
solutions in $\bigcup_{\nu\in C} (x-\alpha)^\nu\bar
C[[x-\alpha]][\log(x-\alpha)]$, for any given $\alpha\in\bar C$.  More
precisely, if $L$ belongs only to $C[x][D]$ and $\alpha\in\bar C$, then these
solutions are actually linear combinations of elements of
$\bigcup_{\nu\in C}(x-\alpha)^\nu C(\alpha)[[x-\alpha]][\log(x-\alpha)]$. 
For a field $K$ with $C\subseteq K\subseteq\bar C$ we will use the notation 
\[
  K[[[x-\alpha]]]:=\bigcup_{\nu\in C}(x-\alpha)^\nu K[[x-\alpha]][\log(x-\alpha)].
\]
Observe that this is not a ring or a $K$-vector space. Also observe that the exponents
$\nu$ are restricted to the small field~$C\subseteq K$, although the dependence on the
choice of $C$ is not reflected by the notation. We hope that the intended field~$C$
will always be clear from the context. 


An operator $L\in\bar C[x][D]$ shall be considered integral if all the terms
in all its series solutions remain above a certain threshold. In the algebraic
case, where series solutions involve at worst only fractional exponents, the
stipulation of having only nonnegative exponents in all the solutions happens
to be equivalent to the requirement that the monic minimal polynomial has
polynomial coefficients. In the differential case, solutions involving
fractional exponents cause factors in the leading coefficient of~$L$
regardless of whether the exponents are positive or negative. Therefore it
doesn't seem promising to use the leading coefficient of~$L$ for defining
integrality.  Instead, we will consider the exponents of its series solutions.
These exponents~$\nu$ may however belong to a field~$C$ which is not necessarily
ordered, and there may be logarithmic terms. For our purposes we will use the
following definition of integrality, depending on a function~$\iota$ which
can be chosen according to the needs of the user.
\begin{defi}\label{def:1}
  Let $\iota\colon C/\set Z\times\set N\to C$ be a function such that
  \begin{enumerate}
  \item $\iota(\nu+\set Z,j)\in\nu+\set Z$ for every $\nu\in C$ and $j\in\set N$, 
  \item $\iota(\nu_1+\set Z,j_1)+\iota(\nu_2+\set Z,j_2)-\iota(\nu_1+\nu_2+\set Z,j_1+j_2)\geq0$
  for every $\nu_1,\nu_2\in C$ and $j_1,j_2\in\set N$,
  \item $\iota(\set Z,0)=0$.
  \end{enumerate}
  A series $f\in \bar C[[[x-\alpha]]]$ is called \emph{integral with respect
    to~$\iota$} if for all terms $(x-\alpha)^\mu\log(x-\alpha)^j$
  occurring with a nonzero coefficient in $f$ we have $\mu-\iota(\mu+\set Z, j)\geq0$.
  (For this to make sense the left-hand sides of the occurring inequalities
  have to be interpreted as integers, not as elements of~$C$.)
\end{defi}

The function $\iota(\cdot,j)$ specifies for each element $\nu+\set Z$ of $C/\set Z$ the smallest
element $\nu$ such that $x^\nu\log(x)^j$ should be considered integral. If
$\iota(\nu+\set Z,j)=\nu$, then $x^\nu\log(x)^j,x^{\nu+1}\log(x)^j$, $\dots$ are integral and
$x^{\nu-1}\log(x)^j,x^{\nu-2}\log(x)^j,\dots$ are not. The condition $\iota(\set Z,0)=0$ implies
that formal Laurent series are integral if and only if they are in fact formal
power series. 

\begin{example}\label{ex:iota}
  A natural choice for $C\subseteq\set C$ is perhaps $\iota(z+\set Z,0)=z$ for
  all $z\in\set C$ with $0\leq\Re(z)<1$, and $\iota(z+\set Z,j)=z$ for all
  $z\in\set C$ with $0<\Re(z)\leq1$ when $j\geq1$. With this convention, a
  term $x^\nu\log(x)^j$ is integral if and only if the corresponding function
  is bounded in a small neighborhood of the origin. For example, $1$,
  $x^{\sqrt{-1}}$, $x\log(x)$ all are integral, while $x^{-1}$,
  $x^{\sqrt{-1}-1}$, $\log(x)$ are not. Unless otherwise stated, we shall
  always assume this choice of $\iota$ in the examples given below.
\end{example}

\begin{prop}\label{prop:int}
  Let $\alpha\in\bar C$ and let $R$ be the set of all $\bar C$-linear combinations of
  series in 
  $(x-\alpha)^\nu\bar C[[x-\alpha]][\log(x-\alpha)]$, $\nu\in C$.
  Then:
  \begin{enumerate}
  \item In every series $f\in R$ there are at most finitely many terms
    $(x-\alpha)^\mu\log(x-\alpha)^j$ which are not integral.
  \item The set $R$ together with the natural addition and multiplication forms
    a ring, and $\{\,f\in R\mid\text{$f$ is integral}\,\}$ forms a subring
    of~$R$.
  \end{enumerate}
\end{prop}
\begin{proof}
  1. \ First consider the case when $f\in (x-\alpha)^\nu\bar
  C[[x-\alpha]][\log(x-\alpha)]$ for some $\nu\in C$. Let $\deg(f)$ denote the
  highest power of $\log(x-\alpha)$ in~$f$.  Then the only possible non-integral
  terms in $f$ can be $(x-\alpha)^{\nu+i}\log(x-\alpha)^j$ for
  $j\in\{0,\dots,\deg(f)\}$ and $i\in\{0,\dots,\iota(\nu+\set
  Z,j)-\nu-1\}$. These are finitely many. In general, if $f$ is a linear
  combination of some series in $(x-\alpha)^\nu\bar
  C[[x-\alpha]][\log(x-\alpha)]$ with possibly distinct~$\nu\in C$, the set of
  all non-integral terms is still a finite union of finite sets of non-integral
  terms, and therefore finite.

  2. \ It is clear that $R$ is a ring. To see that the integral elements form a subring,
  let $f,g\in R$ be integral. Then the series $f+g$ cannot contain any term which is
  not present in at least one of the two summands, so all terms of $f+g$ are integral 
  and $f+g$ as a whole is integral. Now consider multiplication:
  for any term $(x-\alpha)^\mu\log(x-\alpha)^j$ in $f\cdot g$
  there must be some terms $\tau$ in $f$ and $\sigma$ in $g$ such that $\sigma\tau=
  (x-\alpha)^\mu\log(x-\alpha)^j$, say $\tau=(x-\alpha)^{\mu_1}\log(x-\alpha)^{j_1}$
  and $\sigma=(x-\alpha)^{\mu_2}\log(x-\alpha)^{j_2}$. 
  Since $f$ and $g$ are integral, we have $\mu_1-\iota(\mu_1+\set Z,j_1)\geq0$
  and $\mu_2-\iota(\mu_2+\set Z,j_2)\geq0$. The assumption on $\iota$
  in Definition~\ref{def:1} implies that $(\mu_1+\mu_2)-\iota(\mu_1+\mu_2+\set Z,j_1+j_2)
  =\mu-\iota(\mu+\set Z,j)\geq0$. Hence all terms of $f\cdot g$ are integral,
  so also the product of two integral elements is integral.
\end{proof}

\begin{defi}\label{def:monic}
  Let $L\in\bar C(x)[D]$ and $\iota$ be as in Definition~\ref{def:1}.
  \begin{enumerate}
  \item We call $L$ \emph{regular} if it has a fundamental system in
    $\bar C[[[x-\alpha]]]$ for
    every $\alpha\in\bar C$.
  \item $L$ is called \emph{(locally) integral at~$\alpha$ with respect
      to~$\iota$} if it admits a fundamental system in $\bar C[[[x-\alpha]]]$ 
    whose elements all are integral.
  \item $L$ is called \emph{(globally) integral with respect to~$\iota$} if it
    is locally integral at $\alpha$ in the sense of part~2 for every
    $\alpha\in\bar C$.
  \end{enumerate}
\end{defi}

Of course part~2 of this definition is independent of the choice of the
fundamental system. In fact, $L$~is locally integral at $\alpha$ iff all its
series solutions in $x-\alpha$ are integral and form a $\bar C$-vector
space of dimension~$\ord(L)$.

\begin{example}
  \begin{enumerate}
  \item The operator $(2-x)+2(2-2x+x^2)D+4(x-1)xD^2\in\set Q[x][D]$ 
  is locally integral at $\alpha=0$, because its two
  linearly independent solutions
  \begin{alignat*}1
    &1 - \tfrac12x-\tfrac1{24}x^3-\tfrac7{384}x^4-\tfrac{53}{3840}x^5+\O(x^6),\\
    &x^2+\tfrac16x^3+\tfrac16x^4+\tfrac{13}{120}x^5+\O(x^6)
  \end{alignat*}
  are both integral. It is also locally integral at $\alpha=1$, because its two 
  linearly independent solutions 
  \begin{alignat*}1
    &(x{-}1)^{1/2} + \O((x{-}1)^6),\\
    &1-\tfrac12(x{-}1)+\tfrac18(x{-}1)^2-\tfrac1{48}(x{-}1)^3+\O((x{-}1)^4)
  \end{alignat*}
  are integral as well.
  
  The operator is also globally integral because at all $\alpha\in\set C \setminus\{0,1\}$ 
  it has a fundamental system of formal power series, and
  formal power series are always integral.

  \item The operator $1+xD\in\set Q[x][D]$ is not locally integral at $\alpha=0$, because
  it has the non-integral solution~$\frac1x$. It is therefore also not globally integral.
  
  \item The operator $(-1-2x)+(x+2x^2)D+(x^3+x^4)D^2\in\set Q[x][D]$ is
  not locally integral at $\alpha=0$ although all its series solutions are. The
  reason is that it has only one series solution in $\set C[[[x]]]$ 
  while our definition requires that the number of
  linearly independent series solutions must match the order of the operator.
  In other words, generalized series solutions involving exponential terms, like
  the solution $\exp(\frac1x)$ in the present example, are always considered as
  not integral.
\end{enumerate}
\end{example}


Let $L=\ell_0+\cdots+\ell_rD^r\in C[x][D]$ with $\ell_r\neq0$ and consider the
quotient algebra $\bar C(x)[D]/\<L>$, where $\<L>:=\bar C(x)[D]L$ denotes the left ideal
generated by $L$ in~$\bar C(x)[D]$. The algebra $\bar C(x)[D]/\<L>$ is generated as a 
$\bar C(x)$-vector space by the basis $\{1,D,\dots,D^{r-1}\}$. It is also a
$\bar C(x)[D]$-left module, and we can interpret its elements as all those
``functions'' which can be reached by letting an operator $P\in\bar C(x)[D]$ act
on a ``generic solution'' of~$L$, very much like the elements of an algebraic
extension field $\bar C(x)[y]/\<M>$ can be described as those objects which can be
reached by applying a polynomial $P\in\bar C(x)[y]$ to a ``generic root'' of~$M$.  A
difference in this analogy is that in the algebraic case there are only finitely
many roots while in the differential case we have a finite dimensional $\bar C$-vector
space of solutions.

\begin{defi}\label{def:2} 
  Let $L=\ell_0+\cdots+\ell_rD^r\in C[x][D]$ with $\ell_r\neq0$ be a regular
  operator and let $\iota$ be as in Definition~\ref{def:1}.
  \begin{enumerate}
  \item An element $P\in A=\bar C(x)[D]/\<L>$ is called \emph{integral}
    (with respect to~$\iota$)
    if $P\cdot f$ is integral (with respect to~$\iota$)
    for every series solution $f$ of~$L$.
  \item The $\bar C[x]$-left module $\mathcal{O}_L$ of all integral elements of $A$
    is called the \emph{integral closure} of $\bar C[x]$ in~$A$.
  \item A $\bar C(x)$-vector space basis 
    \[\{B_1,\dots,B_r\}\subseteq\bar C(x)[D]/\<L>\]
    is called an \emph{integral basis} if it also generates $\mathcal{O}_L$
    as $\bar C[x]$-left module.
  \end{enumerate}  
\end{defi}

It is easy to see that $\mathcal{O}_L$ is a $\bar C[x]$-left module. Note
however that $\mathcal{O}_L$ is in general not a $\bar C[x][D]$-left module,
because the application of $D$ may turn integral elements into non-integral ones
(for example, $D\cdot x^{1/2}=\tfrac12 x^{-1/2}$ when $\iota(\tfrac12+\set Z,0)=\tfrac12$).

\begin{example}
  \begin{enumerate}
  \item The operator $L=1-D\in\set Q[x][D]$ has for every $\alpha\in\set C$ one
    power series solution of the form $f=1+\O(x-\alpha)$. 
    Since $f$ is integral we have $1\in\mathcal{O}_L$.
    Since $(x-\alpha)^{-1}f$ is not integral for any~$\alpha$, 
    we have in fact that $\{1\}$ is an integral basis.
  \item The operator $L=1+xD$ has the solution $f=\frac1x$. It is
    integral for every $\alpha\neq0$, but not integral at $\alpha=0$. 
    However, $xf=1$ is integral, hence $x\in\mathcal{O}_L$,
    and in fact $\{x\}$ is an integral basis. 
  \item Whenever $L$ has only power series solutions at every $\alpha\in\bar C$, we
    clearly have $\{1,D,\dots,D^{r-1}\}\subseteq\mathcal{O}_L$. However, there
    may still be integral elements that are not $C[x]$-linear combinations of
    these. For example, observe that for the operator $L=(x-1)+D-xD^2$, which
    has two solutions $\exp(x)=1+x+\tfrac12x^2+\O(x^3)$ and $(2x+1)\exp(-x)=x^2+\O(x^3)$,
    we have the nontrivial element $\tfrac1x(1-D)\in\mathcal{O}_L$. Note
    that it is integral at all $\alpha\neq0$ as well.
  \item It can also happen that $1\in\mathcal{O}_L$ but $D\not\in\mathcal{O}_L$.
    For example, for $L=(-1+2x)+(1-4x)D+2xD^2$ we have two solutions
    $1+x+\tfrac12x^2+\O(x^3)$ and $x^{1/2}+x^{3/2}+\tfrac12x^{5/2}+\O(x^3)$ at
    $\alpha=0$. Since both are integral (and there are two linearly independent
    power series solutions for every $\alpha\neq0$) we have $1\in\mathcal{O}_L$.
    However, $D\not\in\mathcal{O}_L$, because the derivative of the second
    solution is $\tfrac12x^{-1/2}+\tfrac32x^{1/2}+\tfrac54x^{3/2}+\O(x^2)$,
    which is not integral since it involves the term $x^{-1/2}$.  An integral
    basis in this case turns out to be $\{1,xD\}$.
  \item We have produced a prototype implementation in Mathematica of the
    algorithm described below. The code is available on the homepage of the first author.
    For the operator $L=x^3D^3+xD-1$, it finds the integral basis
    $\{1,xD,xD^2-D+\frac{1}{x}\}$. A fundamental system of $L$ is
    $\{x, x\log(x), x\log(x)^2\}$.
  \item Let $L=24x^3D^3-134x^2D^2+373xD-450$. This operator has the solutions
    $x^{3/2}$, $x^{10/3}$, and $x^{15/4}$. Our code finds the integral basis
    \[
    \Bigl\{\frac{1}{x},\frac{1}{x^2}D-\frac{3}{2x^3},\frac{1}{x}D^2-\frac{7}{2x^2}D+\frac{9}{2x^3}\Bigr\}.
    \]
  \end{enumerate}
\end{example}

In the analogy with algebraic functions, the integral operators from
Definition~\ref{def:monic} correspond to the monic minimal polynomials with
coefficients in a ring, and the integral elements of Definition~\ref{def:2}
correspond to integral elements of an algebraic function
field. Definitions~\ref{def:monic} and~\ref{def:2} are obviously connected as
follows.

\begin{prop}
  Let $L\in C[x][D]$ and $\tilde L\in\bar C(x)[D]$ be regular and 
  assume that there exists $P\in\bar C(x)[D]$ such that for every $\alpha\in\bar C$
  we have 
  \[
    \{\,f \mid \tilde L\cdot f=0\,\}
   =\{\,P \cdot f\mid L\cdot f=0\,\}
  \]
  where $f$ runs over $\bar C[[[x-\alpha]]]$ on both sides. Then
  $P+\<L>\in\bar C(x)[D]/\<L>$
  is integral in the sense of Definition~\ref{def:2} if and only if $\tilde L$
  is integral in the sense of Definition~\ref{def:monic}.
\end{prop}

\begin{lemma}\label{lem:denom}
  Let $L=\ell_0+\cdots+\ell_rD^r\in \bar C[x][D]$ with $\ell_r\neq0$ be a regular
  operator. Let $p_0,\dots,p_{r-1}\in\bar C(x)$ and let $p=x-\alpha\in\bar C[x]$
  be a factor of the common denominator of $p_0,\dots,p_{r-1}$.
  If $p_0+\cdots+p_{r-1}D^{r-1}\in\mathcal{O}_L$ then $p\mid\ell_r$. 
\end{lemma}
\begin{proof}
  After performing a change of variables, we may assume that $p=x$.  By a
  classical result about linear differential equations (e.g., \cite{ince26}),
  $x\nmid\ell_r$ implies that $L$ admits a fundamental system
  $b_0,\dots,b_{r-1}$ in $C[[x]]$ with $b_i=x^i+\O(x^r)$ for $i=0,\dots,r-1$.
  Then $D^jb_i=i(i-1)\cdots(i-j+1)x^{i-j}+\O(x^{r-j})$ for $i=0,\dots,r-1$ and
  $j=0,\dots,r-1$. Let $e_i$ be the largest integer such that $x^{e_i}$ divides
  the denominator of~$p_i$, let $e=\max\{e_0,\dots,e_{r-1}\}$, and let
  $i\in\{0,\dots,r-1\}$ be some index with $e_i=e$. Then
  $p_iD^ib_i=i!x^{-e}+\O(x^{-e+1})$ and $p_jD^jb_i=\O(x^{-e+1})$ for all $j\neq
  i$. Hence $(p_0+p_1D+\cdots+p_{r-1}D^{r-1})\cdot b_i=i!x^{-e}+\O(x^{-e+1})$
  is not integral because $-e-\iota(-e+\set Z,0)=-e-\iota(\set Z,0)=-e<0$, 
  and hence $p_0+p_1D+\cdots+p_{r-1}D^{r-1}\not\in\mathcal{O}_L$.
\end{proof}


\section{Algorithm Outline}


We shall now discuss how to construct an integral basis $\{B_0,\dots,B_{r-1}\}$
for a given regular operator $L\in C[x][D]$. The key observation is that van
Hoeij's algorithm for computing integral bases for algebraic function fields as
well as the arguments justifying its correctness and termination carry over
almost literally to the present setting. The remainder of this paper therefore
follows closely the corresponding sections of van Hoeij's paper.

The algorithm computes the basis elements $B_0,\dots,B_{r-1}$ in order, at each
stage~$d\in\{0,\dots,r-1\}$ starting with an initial conservative guess for
$B_d$ and refining it repeatedly until an operator $B_d$ is found which
together with $B_0,\dots,B_{d-1}$ generates the $\bar C[x]$-left module
consisting of all the elements of $\mathcal{O}_L$ corresponding to operators of
order~$d$ or less.  Although parts of the calculation take place in the large
field~$\bar C$, it will be shown that the elements $B_i$ in the resulting
integral basis always have coefficients in the small field~$C$, in which the
coefficients of the input operator~$L$ live.

It is not hard to find a suitable $B_0$: For each root $\alpha\in\bar C$ of the
leading coefficient $\ell_r$ of~$L$, compute the first terms of a basis
$\{b_1,\dots,b_r\}$ of solutions in $\bar C[[[x-\alpha]]]$.  Determine the smallest integer $e_\alpha$ such
that $(x-\alpha)^{e_\alpha}b_i$ is integral for every~$i$ according to the
chosen~$\iota$. Then $B_0$ can be set to the product of
$(x-\alpha)^{e_\alpha}$ over all~$\alpha$.  Since $e_\alpha=e_{\tilde\alpha}$
whenever $\tilde\alpha$ is a conjugate of~$\alpha$, it follows that $B_0$
belongs to~$C(x)$.

The outline of the algorithm is now given on a conceptual level.
In Section~\ref{sec:constr} a more detailed description of steps 5--7
will be given.

\def\eatspace#1{#1}
\def\step#1#2{\par\kern1pt\hangindent#2em\hangafter=1\noindent\rlap{\small#1}\kern#2em\relax\eatspace}

\begin{algorithm}\label{alg:8}\em\leavevmode\null\\
INPUT: A regular operator $L=\ell_0+\cdots+\ell_rD^r\in C[x][D]$ with $\ell_r\neq0$\\
OUTPUT: $\{B_0,\dots,B_{r-1}\}\subseteq C(x)[D]/\<L>$, an integral basis 
  of $\bar C(x)[D]/\<L>$.

\smallskip
\step 1 1 Set $s$ to the squarefree part of $\ell_r$. 
\step 2 1 Set $B_0$ to the zero-order operator described above. 
\step 3 1 For $d=1,\dots,r-1$, do the following:
\step 4 2 Set $B_d=s\,D\,B_{d-1}$. (Also $B_d=s^dD^dB_0$ would work.) Consider
\[
  E=\{\,A\in\mathcal{O}_L : \ord(A)\leq d\,\}\setminus \bigl(\bar C[x]B_0 + \cdots + \bar C[x]B_d\bigr).
\]
\step 5 2 While $E\neq\emptyset$, do the following:
\step 6 3 Construct $A\in E$ of the form
\[
  A = \frac1\k\bigl(a_0B_0+\cdots+a_{d-1}B_{d-1}+B_d\bigr)
\]
with $a_0,\dots,a_{d-1},\k\in C[x]$.
\step 7 3  We have
\begin{alignat*}1
  &\bar C[x]B_0+\cdots+\bar C[x]B_{d-1}+\bar C[x]B_d\\
  \subsetneq{}&
  \bar C[x]B_0+\cdots+\bar C[x]B_{d-1}+\bar C[x]A
  \subseteq
  \mathcal{O}_L.
\end{alignat*}
Replace $B_d$ by~$A$. (This makes $E$ strictly smaller.)
\step 8 1  Return $\{B_0,\dots,B_{r-1}\}$.
\end{algorithm}

In the refined version of the algorithm, we will see that the set~$E$ is never
explicitly constructed.  Instead, it suffices to be able to solve the
following subproblems. First, we need to decide whether $E=\emptyset$ for
recognizing the termination of the loop in lines 5--7; this is discussed in
Section~\ref{sec:constr}. Second, we need to show the existence of an
element~$A\in E$ of the form required in step~6 whenever $E\neq\emptyset$; see
Section~\ref{sec:exists}. In Section~\ref{sec:constr} we explain how such
an~$A$ is constructed. Finally, the termination of the loop in lines 5--7 is
proved in Section~\ref{sec:term}.  Except for these issues, the correctness of
the algorithm is obvious.

{\mathversion{bold}
\section{Existence of \texorpdfstring{{\large $A$} if {\large$E\neq\emptyset$}}{A if E is not empty}}\label{sec:exists}}

The arguments in this section are almost identical to those
in~\cite{vanHoeij94}. Nevertheless, for sake of completeness, we
formulate them here for the differential case.

In the $d$-th iteration of the algorithm we can assume by induction 
that $B_0,\dots,B_{d-1}$ with $\ord(B_i)=i$ for all~$i$ form a $\bar C[x]$-left module basis of all integral
elements of order up to $d-1$. We consider the case where the current choice
of~$B_d$, together with $B_0,\dots,B_{d-1}$, does not generate all integral
elements of order up to~$d$, i.e., $E\neq\emptyset$. Recall that
\[
  E=\{\,A\in\mathcal{O}_L : \ord(A)\leq d\,\}\setminus \bigl(\bar C[x]B_0 + \cdots + \bar C[x]B_d\bigr).
\]
We need to show that there exists an integral element $A\in E$ which can be
written in the form $\frac1\k\big(a_0B_0+\cdots+a_dB_d\big)$ with
$a_0,\dots,a_d,\k\in C[x]$ and $a_d=1$. The idea is as follows: starting from
an arbitrary element $A\in E$, we construct, in several steps, simpler
elements in~$E$ until we obtain one with the desired properties.

\begin{lemma}\label{lemma:9}
  If $E\neq\emptyset$, then there exists $A\in E$ of the form
  \begin{alignat}1\label{eq:1}
   A=\frac1{x-\alpha}\bigl(a_0B_0+\cdots+a_{d-1}B_{d-1}+a_dB_d\bigr)
  \end{alignat}
  with $\alpha\in\bar C$, $a_0,\dots,a_{d-1},a_d\in\bar C[x]$.
\end{lemma}
\begin{proof}
  Let $A\in E$, say $A=a_0B_0+\cdots+a_dB_d$ for some $a_i\in\bar
  C(x)$. Since $A\not\in\bar C[x]B_0+\cdots+\bar C[x]B_d$, at least one $a_i$
  must be in $\bar C(x)\setminus\bar C[x]$.  Let $p\in\bar C[x]$ be the common denominator
  of all the~$a_i$, and let $\alpha\in\bar C$ be a root of~$p$.  Then
  $\frac{p}{x-\alpha}A$ has the required form.  To see that it belongs to~$E$,
  notice that $\frac{p}{x-\alpha}\in\bar C[x]$ and
  $\mathcal{O}_L$ is a $\bar C[x]$-module, and that $\frac{p}{x-\alpha}A\not\in
  \bar C[x]B_0+\cdots+\bar C[x]B_d$. 
\end{proof}

\begin{lemma}\label{lemma:10}
  If $A\in E$ and $P\in\bar C[x]B_0+\cdots+\bar C[x]B_d$, then $A+P\in E$.
\end{lemma}
\begin{proof}
  $A\in E\subseteq\mathcal{O}_L$ and $P\in\bar C[x]B_0+\cdots+\bar C[x]B_d\subseteq\mathcal{O}_L$
  implies that $A+P\in\mathcal{O}_L$.
  It is also clear that $\ord(A+P)\leq d$, because $\ord(A)\leq d$ and $\ord(P)\leq d$.
  Finally, to show that $A+P\not\in\bar C[x]B_0+\cdots+\bar C[x]B_d$, assume otherwise.
  Then also $A=(A+P)-P\in\bar C[x]B_0+\cdots+\bar C[x]B_d$ in contradiction to $A\in E$.
\end{proof}

\begin{lemma}\label{lemma:11}
  If $E$ contains an element of the form~\eqref{eq:1}, then it also
  contains such an element with $a_0,\dots,a_{d-1}\in\bar C$ and $a_d=1$.
\end{lemma}
\begin{proof}
  Let  $A=\frac1{x-\alpha}\bigl(a_0B_0+\cdots+a_dB_d\bigr)\in E$ be
  of the form~\eqref{eq:1}. For each $i=0,\dots,d$, write
  $a_i=(x-\alpha)p_i+a'_i$ with $p_i\in\bar C[x]$ and  $a'_i\in\bar C$. By
  Lemma~\ref{lemma:10}, $A\in E$ implies $A'\in E$ for
  $A':=\frac1{x-\alpha}\bigl(a'_0B_0+\cdots+a'_{d-1}B_{d-1}+a'_dB_d\bigr)$. Since
  $B_0,\dots,B_{d-1}$ are assumed to generate the submodule of all the
  elements of $\mathcal{O}_L$ of order at most~$d-1$, we have $a'_d\neq0$.
  Dividing $A'$ by $a'_d$ yields an element of $E$ of the requested form.
\end{proof}

\begin{lemma}\label{lemma:12}
  If $E$ contains an element of the form~\eqref{eq:1}
  with $a_0,\dots,a_{d-1}\in\bar C$ and $a_d=1$, then it also contains
  such an element with $a_0,\dots,a_{d-1}\in C(\alpha)$ and $a_d=1$.
\end{lemma}
\begin{proof}
  Let $A\in E$ be of the form~\eqref{eq:1} with $a_0,\dots,a_{d-1}\in\bar C$ and $a_d=1$.
  Since $\bar C$ is necessarily a $C(\alpha)$-vector space, there are some
  $C(\alpha)$-linearly independent elements $e_0,\dots,e_n$ of $\bar C$ such
  that $a_0,\dots,a_d$ all belong to $V=e_0 C(\alpha)+\cdots+ e_n C(\alpha)$. We
  may assume $e_0=1$. Consider a fundamental system $b_1,\dots,b_r\in C(\alpha)[[[x-\alpha]]]$ 
  of~$L$. Then
  each $A\cdot b_j$ has coefficients in $V$ and, since $A\in
  E\subseteq\mathcal{O}_L$, only involves integral terms. For an element $v\in V$
  let us write $[e_i]v$ for the coordinate of~$v$ with respect to~$e_i$. By the linear
  independence of the $e_i$ over $C(\alpha)$, the series $[e_i]\bigl(A\cdot
  b_j\bigr)=([e_i]A)\cdot b_j$ obtained from $A\cdot b_j$ by replacing each
  coefficient by its $e_i$-coordinate will be
  integral.
  In particular, the operator $A_0=[e_0]A\in C(\alpha)[x][D]$ must belong
  to~$E$. Because of $[e_0]a_d=[e_0]1=1$, it meets all the requirements. 
\end{proof}

\begin{lemma}\label{lemma:13}
  If $E$ contains an element of the form~\eqref{eq:1} with
  $a_0,\dots,a_{d-1}\in C(\alpha)$ and $a_d=1$, then it also contains
  such an element with $a_0,\dots,a_{d-1}\in C[x]$ and $a_d=1$.
\end{lemma}
\begin{proof}
  For every $n>0$ we have $x-\alpha\mid x^n-\alpha^n$ in
  $\bar C[x]$, and thus also $x-\alpha\mid p(x)-p(\alpha)$ for $p\in\bar
  C[x]\setminus\bar C$. Therefore, if we view the $a_i\in C(\alpha)$ as
  polynomials in~$\alpha$, then replacing $\alpha$ in them by $x$ amounts to
  adding some polynomial multiple of $(x-\alpha)$ to them. This change means for
  $A=\frac1{x-\alpha}(a_0B_0+\cdots+a_{d-1}B_{d-1}+B_d)$ that adding a suitable
  element $P\in
  C(\alpha)[x]B_0+\cdots+C(\alpha)[x]B_{d-1}\subseteq\mathcal{O}_L$ turns $A$
  into an operator of the requested form. 
  By Lemma~\ref{lemma:10}, this new operator also belongs to~$E$.
\end{proof}

\begin{theorem}
  If $E\neq\emptyset$, then there exists an element $A\in E$ of the form
  \[
    A=\frac1\k\bigl(a_0B_0+\cdots+a_{d-1}B_{d-1}+B_d\bigr)
  \]
  with $\k\in C[x]$ an irreducible factor of $\ell_r$ and
  $a_0,\dots,a_{d-1}\in C[x]$ such that $\deg(a_i)<\deg(\k)$ for all~$i$.
\end{theorem}
\begin{proof}
  The assumption $E\neq\emptyset$ in combination with Lemmas~\ref{lemma:9},
  \ref{lemma:11}, \ref{lemma:12}, and~\ref{lemma:13} implies that $E$ contains
  an element of the form~\eqref{eq:1} with $a_0,\dots,a_{d-1}\in C[x]$ and
  $a_d=1$. Furthermore, Lemma~\ref{lem:denom} implies that $\alpha$ is a root
  of~$\ell_r$. Let $\k\mid\ell_r$ be the minimal polynomial of~$\alpha$. We
  claim that $A:=\frac1\k B\in E$ where $B:=a_0B_0+\cdots+a_{d-1}B_{d-1}+B_d$.

  To prove this, we have to show that for every $\tilde\alpha\in\bar C$ and
  every solution $\tilde b\in C(\tilde\alpha)[[[x-\tilde\alpha]]]$ of~$L$ we still have
  that $A\cdot\tilde b$ is integral.  When $\tilde\alpha$ is not a root of~$\k$,
  this is clear because $1/\k$ admits an expansion in $C[[x-\tilde\alpha]]$, and
  multiplication of the integral series $B\cdot\tilde b$ by
  a formal power series preserves integrality by Proposition~\ref{prop:int}. When
  $\tilde\alpha=\alpha$, write $\k=(x-\alpha)q$ for some $q\in\bar C[x]$ with
  $x-\alpha\nmid q$ and note that $1/q$ admits an expansion in $\bar
  C[[x-\alpha]]$ and $\frac1{x-\alpha}B\cdot\tilde b$ is
  integral, so $\frac1\k B\cdot\tilde b$ is integral
  too. When $\tilde\alpha$ is a conjugate of~$\alpha$, note that
  $\frac1{x-\tilde\alpha}B\cdot\tilde b$ must be integral,
  because if it were not, then for the series $b\in C(\alpha)[[[x-\alpha]]]$ obtained from
  $\tilde b$ via the conjugation map that sends $\tilde\alpha$ to $\alpha$ we
  would have that $\frac1{x-\alpha}B\cdot b$ is also not
  integral, in contradiction to our choice of $a_0,\dots,a_d$. Therefore the
  same argument as in the case $\tilde\alpha=\alpha$ applies.

  This completes the proof of the claim. To complete the proof of the theorem,
  note that the claimed degree bounds on $a_i$ can be ensured by
  Lemma~\ref{lemma:10}.
\end{proof}

{\mathversion{bold}
\section{Construction of \texorpdfstring{{\large$A$}}{A} in Step~6}\label{sec:constr}}

In the previous section we have demonstrated that in step~6 of the algorithm
it suffices to search for an integral element~$A$ of the form
\[
  A = \frac1\k\bigl(a_0B_0+\cdots+a_{d-1}B_{d-1}+B_d\bigr)
\]
where $a_0,\dots,a_{d-1},\k\in C[x]$, $\k\mid\ell_r$ and $\deg(a_i)<\deg(p)$.
Conversely, this means that if no such $A$ exists, the set $E$ is empty. 

For each irreducible factor~$\k$ of $\ell_r$ one can set up an ansatz for~$A$
with undetermined coefficients $a_0,\dots,a_{d-1}$. We want to find
$a_0,\dots,a_{d-1}$ such that $A\cdot f$ is integral for all solutions~$f$
of~$L$. Note that we need to enforce integrality only for series solutions in
$x-\alpha$ where $\alpha$ is a root of~$\k$.  Choosing a fundamental system
$b_1,\dots,b_r$ of such solutions, computing the first terms of $B_j\cdot b_i$,
plugging them into the ansatz, and equating the coefficients of all non-integral
terms to zero yields a linear system for $a_0,\dots,a_{d-1}$. If this system
does not admit a solution, one knows that no such~$A$ with denominator~$\k$
exists.

In summary, the loop in lines 5--7 of Algorithm~\ref{alg:8} can be described in
more detail as follows.

\medskip 
\step {5a} 2 Let $Q\subseteq\bar C$ be a set containing exactly one
root $\alpha\in\bar C$ for each irreducible factor $\k$ of~$\ell_r$.  
\step {5b} 2 While $Q\neq\emptyset$, do the following: 
\step {5c} 3 For all $\alpha\in Q$, do the following:

\smallskip
\step {6a} 4 Let $b_1,\dots,b_r$ be a fundamental system of $L$ 
  in ${C(\alpha)}[[[x-\alpha]]]$. 
\step {6b} 4 With variables $a_0,\dots,a_{d-1}$, form the series 
\[
  \bigl(a_0B_0 + \cdots + a_{d-1}B_{d-1} + B_d\bigr)b_i
\]
for $i=1,\dots,r$. 
\step {6c} 4 Construct a linear system for $a_0,\dots,a_{d-1}$ by equating the
coefficients of all the non-integral terms in these series to zero. 

\smallskip
\step {7a} 4 If the system has a solution $(a_0,\dots,a_{d-1})\in C(\alpha)^d$:
\step {7b} 5 Let $\k$ be the minimal polynomial of $\alpha$ over~$C$.
\step {7c} 5 Replace each $a_i\in C(\alpha)=C[x]/\<\k>$ by the corresponding 
polynomial in $C[x]$ of degree less than $\deg(\k)$.
\step {7d} 5 Replace $B_d$ by $\frac1\k(a_0B_0+\cdots+a_{d-1}B_{d-1}+B_d)$.
\step {7e} 4 Otherwise
\step {7f} 5 discard $\alpha$ from~$Q$.

\medskip

Despite being more detailed than the listing given in Algorithm~\ref{alg:8},
these lines are still somewhat conceptual. An actual implementation cannot just
``let'' $b_i$ be some infinite series object, and it does not need to. What we
need are only the terms of $b_i$ that give rise to some non-integral terms of
$\bigl(a_0B_0 + \cdots + a_{d-1}B_{d-1} + B_d)b_i$. These are only finitely
many by Proposition~\ref{prop:int}. In Section~\ref{sec:bounds} we address the
question how many terms of~$b_i$ we need to compute.

\section{Termination}\label{sec:term}

\def\wr{\operatorname{wr}} 

The termination of van Hoeij's algorithm~\cite{vanHoeij94} is established by
the observation that the degree of a certain polynomial, starting with the
discriminant $\operatorname{Res}_y\big(M,\frac{\partial M}{\partial y}\big)$,
decreases in each iteration of the main loop. In the case of D-finite functions,
the role of the discriminant is played by the \emph{Wronskian} and a
generalized version of it. Recall that the Wronskian of the functions
$f_1(x),\dots,f_r(x)$ is defined as the determinant
\begin{equation}\label{eq:wronsk}
  W = \begin{vmatrix}
    f_1(x) & f_2(x) & \cdots & f_r(x) \\
    f_1'(x) & f_2'(x) & \cdots & f_r'(x) \\
    \vdots & \vdots & \ddots & \vdots \\
    f_1^{(r-1)}(x) & f_2^{(r-1)}(x) & \cdots & f_r^{(r-1)}(x)
  \end{vmatrix}.
\end{equation}

\begin{defi}
Let $L\in\bar C[x][D]$ be regular and let
$b_1,\dots,b_r$ be a fundamental system of~$L$ in $\bar C[[[x-\alpha]]]$ 
for some $\alpha\in\bar C$. For $B_0,\dots,B_{r-1}\in\bar C(x)[D]/\<L>$ we define the
\emph{generalized Wronskian at~$\alpha$}, as 
\[
  \wr_{L,\alpha}(B_0,\dots,B_{r-1}) :=
  \begin{vmatrix}
    B_0\cdot b_1 & \cdots & B_0\cdot b_r \\
    \vdots & \ddots & \vdots \\
    B_{r-1}\cdot b_1 & \cdots & B_{r-1}\cdot b_r
  \end{vmatrix}.
\]
\end{defi}

Note that the generalized Wronskian $\wr_{L,\alpha}(B_0,\dots,B_{r-1})$
belongs to $\bar C[[[x-\alpha]]]$ and that the choice of a different fundamental system instead of
$b_1,\dots,b_r$ only changes its value by a nonzero multiplicative constant,
which will be irrelevant for our purpose.

For the special choice $B_i=D^i$, the generalized Wronskian
$\wr_{L,\alpha}(1,D,\dots,D^{r-1})$ reduces to the Wronskian~\eqref{eq:wronsk}
with $f_i=b_i$. It is well-known and easy to check that the classical
Wronskian~\eqref{eq:wronsk} of $b_1,\dots,b_r$ satisfies the first-order
equation $\ell_rD\,W+\ell_{r-1}W=0$ and hence is hyperexponential. Since the
generalized Wronskian can be obtained from the usual Wronskian by elementary
row operations over~$C(x)$, it is clear that also the generalized Wronskian
is hyperexponential.


\begin{theorem}
Algorithm~\ref{alg:8} terminates.
\end{theorem}
\begin{proof}
First observe that during the whole execution of the algorithm,
$B_0,\dots,B_{r-1}\in C(x)[D]/\<L>$ are integral, i.e., $B_0\cdot
f,\dots,B_{r-1}\cdot f$ are integral for any series solution~$f$ of~$L$
according to Definition~\ref{def:2}. (Actually, the $B_d$'s are constructed
one after the other, but they can be initialized with $B_d=s^dD^dB_0$.)  This
means that, at any time and for any $\alpha\in\bar C$, the generalized
Wronskian $\wr_{L,\alpha}(B_0,\dots,B_{r-1})$ is integral, as it is the sum of
products of integral series (see Proposition~\ref{prop:int}). Since it is
hyperexponential, it follows that it has no logarithmic terms.  Every nonzero
term of $\wr_{L,\alpha}(B_0,\dots,B_{r-1})$ is therefore of the form
$(x-\alpha)^{\mu}$ with $\mu=\iota(\mu+\set Z,0)+m$ for some nonnegative
integer~$m$. For each $\alpha\in\bar C$ let $m_\alpha$ be the smallest such
integer. Now let $n=\sum_{\alpha\in Q}m_\alpha$ where $Q$ is defined
  as in step 5a.  Each time $B_d$ is updated in the algorithm (either
  in step~4 or in step~7d), none of the $m_\alpha$ can increase and 
  exactly one of them strictly decreases, so also $n$ decreases. More
precisely, if for example $B_d$ is replaced by
$\frac1\k\big(a_0B_0+\cdots+a_{d-1}B_{d-1}+B_d\big)$ in step~7, then
$\wr_{L,\alpha}(B_0,\dots,B_d)$ is divided by~$\k$ (recall that $\k$ is a
non-constant polynomial in $C[x]$). But the $m_\alpha$ cannot become negative
as this would violate the integrality of
$\wr_{L,\alpha}(B_0,\dots,B_{r-1})$. Therefore the algorithm must terminate.
\end{proof}

\section{Bounds}\label{sec:bounds}

In the algebraic case, van Hoeij~\cite{vanHoeij94} gave a-priori bounds on
the orders to which the $b_i$ have to be calculated. His algorithm computes their
terms at the very beginning once and for all in order to avoid their
recomputation inside the loop. He also suggested that the terms of $B_j\cdot
b_i$ for $j<d$ should not be recomputed but cached.

Nowadays, in an object-oriented programming environment, the algorithm can be
implemented in such a way that recomputations of series terms are avoided even
when no a-priori bound on the truncation order is available, via the paradigm of
lazy series~\cite{BurgeWatt89,vanderHoeven02}.

Nevertheless it is desirable to have a-priori bounds available also in the
D-finite case. A rough bound follows immediately from the discussion in
Section~\ref{sec:term}: as we have seen, the Wronskian
$\wr_{L,\alpha}\big(B_0,sDB_0,\dots,s^{r-1}D^{r-1}B_0\big)$ gives a
denominator bound for the elements of the integral basis. More refined bounds
are elaborated in the following.

Let $\alpha\in\bar C$ be a root of the leading coefficient~$\ell_r$ and
$\{b_1,\dots,b_r\}\subset C(\alpha)[[[x-\alpha]]]$ be a fundamental system
of~$L$:
\begin{equation}\label{eq:bi}
  b_i = \sum_{k=0}^{\infty} b_{i,k}\big(\log(x-\alpha)\big)\,(x-\alpha)^{\nu_i+k},
  \quad b_{i,0}\neq 0,
\end{equation}
where $b_{i,k}\in C(\alpha)[\log(x-\alpha)]$ are polynomials in
$\log(x-\alpha)$ such that for each~$i$ the degrees of $b_{i,0},b_{i,1},\dots$
are bounded by some integer~$d_i$.  According to step 5c, we have to consider
each $\alpha\in Q$ separately, so for the rest of this section we fix such
an~$\alpha$.

In step 6a we want to replace $b_1,\dots,b_r$ by truncated series
$t_1,\dots,t_r$ of the form
\begin{equation}\label{eq:ti}
  t_i = \sum_{k=0}^{N_i} b_{i,k}\big(\log(x-\alpha)\big)\,(x-\alpha)^{\nu_i+k}\
  \text{with}\ N_i\in\set N.
\end{equation}
The bounds $N_i$ must be chosen such that this replacement does not change the
result of the algorithm.  The only critical step is when $b_1,\dots,b_r$ are
used to test the integrality of certain elements from the algebra
$C(x)[D]/\<L>$, which are not known in advance. Theorem~\ref{thm.bounds} below gives
a sufficient condition that allows us to use $t_i$ instead of~$b_i$ in the
integrality test, by asserting that its answer does not change, whatever
element of $C(x)[D]/\<L>$ we consider. For brevity, let $R$ denote the ring
$C(\alpha)[[x-\alpha]][\log(x-\alpha)]$ in the subsequent reasoning.

\begin{lemma}\label{lem:m}
Let $\{b_1,\dots,b_r\}\subset C(\alpha)[[[x-\alpha]]]$ be a fundamental system
of the form~\eqref{eq:bi} with $\nu_i$ as above, and let $W_b=(D^j\cdot b_i)_{1\leq i\leq
  r,0\leq j<r}$. Then we can find an $m\in\set N$ such that
\[
  \det(W_b)=\sum_{k=0}^\infty w_k\,(x-\alpha)^{\nu_1+\dots+\nu_r-r(r-1)/2+m+k}
\]
with $w_0\neq0$.
\end{lemma}
\begin{proof}
For the $(i,j)$-entry of $W_b$ we have
\[
  (W_b)_{i,j} = D^{j-1}\cdot b_i \in (x-\alpha)^{\nu_i-j+1}R
\]
and therefore
\[
  \det(W_b) \in (x-\alpha)^{\nu_1+\cdots+\nu_r-r(r-1)/2}R.
\]
Note that $\det(W_b)\neq0$ because it is precisely the Wronskian of
$b_1,\dots,b_r$. It follows that a unique $m\geq0$ with the desired
property exists.
\end{proof}

\begin{theorem}\label{thm.bounds}
Let $L\in C(x)[D]$ be an operator of order~$r$ and $\{b_1,\dots,b_r\}\subset
C(\alpha)[[[x-\alpha]]]$ be a fundamental system of~$L$ with $\nu_i$ and $d_i$
as above. Moreover, let $m\in\set N$ be as in Lemma~\ref{lem:m} and let
$N_1,\dots,N_r\in\set N$ be given by
\[
  N_i = m\, +\!\! \max_{\genfrac{}{}{0pt}{1}{1\leq j\leq r}{0\leq k<d_i+r}} \Big(\iota(\nu_i-\nu_j+\set Z,k) -(\nu_i-\nu_j)\Big).
\]
If $t_i$ is the truncation~\eqref{eq:ti} of $b_i$ at order~$N_i$, for $1\leq
i\leq r$, then for all $B\in C(x)[D]/\<L>$ we have the equivalence:
\begin{equation}\label{eq:equiv}
  \forall i\colon B\cdot b_i\ \text{is\ integral}\ \iff
  \forall i\colon B\cdot t_i\ \text{is\ integral}.
\end{equation}
\end{theorem}
\begin{proof}
We introduce the matrix $W_b=(D^j\cdot b_i)_{1\leq i\leq r,0\leq j<r}$ as
before, and the short notation $B\cdot b=(B\cdot b_1,\dots,B\cdot
b_r)$. Analogously we define $W_t$ and $B\cdot t$. A vector resp. matrix is
called integral if all its entries are integral.  If $c$ is the coefficient
vector of~$B$, i.e., $c\cdot(1,D,\dots,D^{r-1})=B$, then we have $B\cdot b =
W_b\,c$ and $B\cdot t = W_t\,c$.  Combining these two equations we get
\begin{equation}\label{eq:wtwb1}
  B\cdot t=W_tW_b^{-1}\,(B\cdot b).
\end{equation}
Setting $Z=W_b-W_t$ yields
\begin{equation}\label{eq:zwb1}
  W_tW_b^{-1}=\mathrm{Id}_r-ZW_b^{-1}.
\end{equation}

The proof is split into two parts, according to the two directions
of the equivalence~\eqref{eq:equiv}.

\emph{Part 1:} If we assume that $B\cdot b$ is integral, then \eqref{eq:wtwb1}
exhibits that the integrality of $W_tW_b^{-1}$ is a sufficient condition to
conclude that also $B\cdot t$ is integral, using Proposition~\ref{prop:int}.
By \eqref{eq:zwb1} it suffices to show that $ZW_b^{-1}$ is integral.  First of
all we have to argue that $W_b^{-1}\in C(\alpha)[[[x-\alpha]]]^{r\times r}$
since otherwise Definition~\ref{def:1} would not be applicable. In
Section~\ref{sec:term} we have remarked that the Wronskian $\det(W_b)$ is
hyperexponential.  In particular, it involves no logarithmic terms and
therefore is invertible in $C(\alpha)[[[x-\alpha]]]$.  Using Cramer's rule we
find that
\[
  \big(W_b^{-1}\big)_{i,j} = (-1)^{i+j} \frac{\det W_b^{[j,i]}}{\det W_b} \in (x-\alpha)^{i-\nu_j-m-1}R,
\]
where $W_b^{[j,i]}$ is the matrix obtained by deleting row~$j$ and column~$i$
from~$W_b$. So the entries of $W_b^{-1}$ are series in $C(\alpha)[[[x-\alpha]]]$.
The fact that $\det W_b^{[j,i]}$ satisfies a differential equation of order
less than or equal to~$r$ implies that the highest power of $\log(x-\alpha)$
that can appear in the entries of $W_b^{-1}$ is~$r-1$. On the other hand, it
is easy to see that $Z_{i,j}\in(x-\alpha)^{\nu_i+N_i-j+2}R$, so it follows
that
\begin{equation}\label{eq:v1}
  \big(ZW_b^{-1}\big)_{i,j}\in(x-\alpha)^{\nu_i-\nu_j+N_i-m+1}R,
\end{equation}
and that herein $\log(x-\alpha)$ appears with exponent at most $d_i+r-1$. By our
choice of~$N_i$ the series in~\eqref{eq:v1} is integral for all
$1\leq i,j\leq r$ and therefore the whole matrix $ZW_b^{-1}$.

\emph{Part 2:} Now assume that $B\cdot b$ is not integral.  Then from
\[
  B\cdot t = \big(\mathrm{Id}_r-ZW_b^{-1}\big)(B\cdot b) = B\cdot b - \big(ZW_b^{-1}\big)(B\cdot b)
\]
it follows that $B\cdot t$ is non-integral as well. To see this, let $n$ be
the largest integer such that a term of the form $(x-\alpha)^{\iota(\mu+\set
  Z,k)-n}\log(x-\alpha)^k$ appears in $B\cdot b$ for some $\mu\in C$ and
$k\in\set N$.  Let $i$ be an index such that a term of the given form appears
in $B\cdot b_i$ with nonzero coefficient. This term cannot be canceled in
\[
  B\cdot t_i = B\cdot b_i - \sum_{j=1}^r \big(ZW_b^{-1}\big)_{i,j} \big(B\cdot b_j\big)
\]
because all terms of the series $(ZW_b^{-1}\big)_{i,j}$ are of the form
$(x-\alpha)^{\iota(\nu_i-\nu_j+\set Z,k)+\ell}\log(x-\alpha)^k$ with
$\ell\geq1$ by our choice of~$N_i$. So also $B\cdot t$ is not integral.
\end{proof}

\section{Comparison with the\texorpdfstring{\hfill\break}{} algebraic case}

We have shown that the underlying ideas of van Hoeij's algorithm for computing
integral bases of algebraic function fields apply in a more general context.
Indeed, it is fair to regard van Hoeij's algorithm as a special case of our
algorithm, since every algebraic function is also D-finite. Recall that an
algebraic function field $C(x)[y]/\<M>$ with some irreducible polynomial $M$ of
degree~$d$ becomes a differential field if we set $D\cdot c=0$ for all $c\in C$,
$D\cdot x=1$, and
\[
  D\cdot y:=-\frac{\frac d{dx}M}{\frac d{dy}M}\bmod M.
\] 
Since $C(x)[y]/\<M>$ is also a
$C(x)$-vector space of dimension~$d$, it is clear that any $d+1$ elements must
be $C(x)$-linearly dependent. This implies the existence of an operator $L\in
C(x)[D]$ of order at most $d$ with $L\cdot y=0$. Usually there is no such
operator of lower order, which means that $y,D\cdot y,\dots,D^{d-1}\cdot y$ are
$C(x)$-linearly independent and thus a basis of $C(x)[y]/\<M>$. In this case, a
vector space basis $\{B_1,\dots,B_d\}\subseteq\set C(x)[y]/\<L>$ is an integral
basis in the sense of Definition~\ref{def:2} if and only if $\{B_1\cdot
y,\dots,B_d\cdot y\}\subseteq\set C(x)[y]/\<M>$ is an integral basis of the
algebraic function field in the classical sense.

When $y\in C(x)[y]/\<M>$ is annihilated by an operator $L$ of order less
than~$d$, we can compute the minimal-order operators $L_0,\dots,L_{d-1}$ which
annihilate $y^0,\dots,y^{d-1}$, respectively, and take
$L=\lclm(L_0,\dots,L_{d-1})$. Then the $C(x)$-vector space generated by all solutions
of $L$ is the whole field $C(x)[y]/\<M>$, and if 
$\{B_1,\dots,B_n\}$ is an integral basis for~$L$, then $\{B_i\cdot
y^j:i=1,\dots,n, \ j=0,\dots,d-1\}$ generates the $C[x]$-module of all integral
elements of $C(x)[y]/\<M>$.

As a less brutal approach, we can simply replace $y$ by some other generator of
the field. In practice, most field generators will have an annihilating operator
of order~$d$, but none of smaller order.

\begin{example}
  An integral basis for the field $\set Q(x)[y]/\<M>$ with
  $M=y^3-x^2$ is $\big\{1,y,\frac1x y^2\big\}$. The lowest-order differential operator
  annihilating $y$ is $L=3xD-2$, which is not useful because its order is less than 
  the degree of~$M$.

  Instead, let us try $Z=1+y+y^2$ as generator. We have $\set Q(x)[y]/\<M>=\set
  Q(x)[Z]/\<N>$, where $N=Z^3-3Z^2-3(x^2-1)Z-x^4+2x^2-1$ is the minimal
  polynomial of~$Z$.  Given $N$ instead of~$M$ as input, van Hoeij's algorithm
  finds the following integral basis for $\set Q(x)[Z]/\<N>$:
  \begin{equation}\label{eq:ib1}
    \left\{1,Z,\frac{Z^2}{x(x-1)(x+1)}-\frac{(x^2+2)Z}{x(x-1)(x+1)}-\frac{1}{x}\right\}.
  \end{equation}
  The lowest order annihilating operator of $Z$ is $L=9x^2D^3+9xD^2-D$. It has the right
  order and our Mathematica implementation returns the integral basis
  $\bigl\{1,xD,xD^2+\frac{1}{3}D\bigr\}$. 
  We can express its derivatives as polynomials in~$Z$, using
  \[
    D\cdot Z = \frac{-2Z^2+2(2x^2+1)Z}{3x(x-1)(x+1)},
  \]
  and obtain the following integral
  basis for 
  $\set Q(x)[Z]/\<N>$:
  \[
    \left\{Z,\frac{-2Z^2+2(2x^2+1)Z}{3(x-1)(x+1)},
      \frac{8(-Z^2+(x^2+2)Z+x^2-1)}{9x(x-1)(x+1)}\right\}.
  \]
  Applying a change of basis with the unimodular matrix
  \[
  \frac18\begin{pmatrix}
    8 & -12 & 9x \\ 8 & 0 & 0 \\ 0 & 0 & -9
  \end{pmatrix}
  \]
  gives the integral basis~\eqref{eq:ib1} computed by Maple.
\end{example}


One of the features of integral bases for algebraic function fields is that they
allow an extension of the classical Hermite reduction for integration of
rational functions to the case of algebraic functions. This was observed by
Trager~\cite{Trager84}. In order to make this work, Trager requires that both the
integral basis as well as the integrand should be ``normal at infinity''. This
corresponds to the condition in the rational case that the rational function to
be integrated must not have a polynomial part. Trager shows that normality of
the integrand can always be achieved by applying a suitable change of variables,
and he gives an algorithm that turns an arbitrary integral basis into one that
is normal at infinity. After that, the Hermite reduction process looks very similar
to the rational case. We give here an example for a non-algebraic D-finite 
function. 

\begin{example}
  Let $L=(2x+1)-(4x^2+1)D+2(2x-1)xD^2$; its solutions are $\exp(x)$ and $\sqrt{x}$,
  but we will not use this information. Let us just write $y$ for a solution of~$L$. 
  An integral basis of $\mathcal{O}_L$ is given by $\{1,\frac{1}{2x-1}(2xD-1)\}$.
  Let $\omega_0:=y$ and $\omega_1:=\frac{1}{2x-1}(2xD-1)\cdot y$ and consider the
  function
  \[
    f=\frac{a_0\omega_0+a_1\omega_1}{uv^m}
  \]
  where $a_0=4x^2+37x-11$,
  $a_1=-28x^3+40x^2-x-1$, 
  $u=4$, $v=(x-1)x$, $m=2$.

  Hermite reduction consists in finding $b_0,b_1,c_0,c_1\in\set Q[x]$ 
  with
  \[
    \frac{a_0\omega_0+a_1\omega_1}{uv^m}
    = \Bigl(\frac{b_0\omega_0+b_1\omega_1}{v^{m-1}}\Bigr)' + \frac{c_0\omega_0+c_1\omega_1}{uv^{m-1}}.
  \]
  After working out the differentiation, multiplying by~$uv^m$, and taking the whole equation mod~$v$
  we are left with the constraint
  \[
    a_0\omega_0+a_1\omega_1 \equiv b_0 u v^m\Bigl(\frac{\omega_0}{v^{m-1}}\Bigr)'
                                 + b_1 u v^m\Bigl(\frac{\omega_1}{v^{m-1}}\Bigr)'
    \bmod v
  \]
  For the derivatives of $\omega_0$ and $\omega_1$ we have
  \[
     D\omega_0=\frac1{2x}\omega_0 - \frac{1-2x}{2x}\omega_1,\quad
     D\omega_1=\omega_1,
  \]
  so that the previous constraint can be rewritten to
  \[
    a_0\omega_0+a_1\omega_1 \equiv
    -\tfrac12
    b_0 u \bigl(3\omega_0+\omega_1)
    -2 b_1 u \omega_1\bmod v.
  \]
  Plugging in $a_0,a_1$ and~$u$ and comparing coefficients of $\omega_i$ 
  leads to the linear system
  \[
    \begin{pmatrix}
      41x-11
      \\
      11x-1
    \end{pmatrix}
    =
    \begin{pmatrix}
      2-6x & 2-2x  \\
      0 & 4-8x
    \end{pmatrix}\begin{pmatrix}
      b_0
      \\
      b_1
    \end{pmatrix}\bmod v
  \] 
  which has the solution $b_0=\tfrac12(4x+11)$, $b_1=\tfrac52(2x-1)$. Next we 
  find that
  \[
    f - \Bigl(\frac{b_0\omega_0+b_1\omega_1}{v^{m-1}}\Bigr)'=\frac{c_0\omega_0+c_1\omega_1}{uv^{m-1}}
  \]
  for $c_0=0$, $c_1=0$. Consequently, we have found that
  \[
    \int f = \frac{(11+4x)\omega_0+5(2x-1)\omega_1}{8(1-x)^2x^2}
           = \frac5{x-1}y' - \frac{2x+3}{(x-1)x}y.
  \]
  The same answer could have been found using an algorithm of Abramov and van
  Hoeij~\cite{abramov99}, using a different approach.
\end{example}

\bibliographystyle{plain}
\bibliography{integral}

\end{document}